\documentclass[a4paper,USenglish,numberwithinsect]{lipics-v2019}

\usepackage{algorithm, algorithmicx, xspace, import}
\usepackage[utf8]{inputenc}

\makeatletter
\g@addto@macro{\UrlBreaks}{\UrlOrds} % allow line break on "-" in URLs
\makeatother

\algblockdefx[PROCEDURE]{Procedure}{EndProcedure}[3][]{\textbf{procedure} \textsc{#2}(#3)\ifthenelse{\equal{#1}{}}{}{\Comment{#1}}\\\textbf{begin}}{\textbf{end}}
\algloopdefx{If}[1]{\textbf{if} #1 \textbf{then}}
\newcommand{\Return}{\textbf{return}\ }

\algblockdefx[WHILE]{While}{EndWhile}[1]{\textbf{while} #1 \textbf{do begin}}{\textbf{end}}
\algblockdefx[DOWHILE]{DoBegin}{EndDoWhile}{\textbf{do begin}}[1]{\textbf{end while} #1}

\newcommand{\Win}{\mathit{Win}}
\newcommand{\Atr}{\textsc{Atr}}
\newcommand{\limsupP}{\textup{Limsup}\ensuremath{P}\xspace}
\newcommand{\limsupeven}{\textup{Limsup\-Even}\xspace}
\newcommand{\limsupodd}{\textup{Limsup\-Odd}\xspace}
\newcommand{\safety}{\textup{Safety}\xspace}
\newcommand{\nodes}{\mathsf{nodes}}
\theoremstyle{theorem}
\newtheorem{fact}[theorem]{Fact}

\hideLIPIcs\nolinenumbers
\title{Parity Games: Zielonka's Algorithm in Quasi-Polynomial Time}

\author{Paweł Parys}{Institute of Informatics, University of Warsaw, Poland}{parys@mimuw.edu.pl}{https://orcid.org/0000-0001-7247-1408}{}

\authorrunning{P. Parys}

\Copyright{Paweł Parys}



\funding{Work supported by the National Science Centre, Poland (grant no.\ 2016/22/E/ST6/00041).}

\acknowledgements{The author would like to thank Wojciech Czerwiński, Laure Daviaud, Marcin Jur\-dziń\-ski, Eryk Kopczyński, Ranko Lazi\'c, Karoliina Lehtinen, and Igor Walukiewicz 
	for all the discussions that preceded writing of this paper, 
	and the anonymous reviewers of the previous version of this paper for their valuable comments.}

\begin{document}

\maketitle

\begin{abstract}
	Calude, Jain, Khoussainov, Li, and Stephan~(2017) proposed a quasi-polynomial-time algorithm solving parity games.
	After this breakthrough result, a few other quasi-polynomial-time algorithms were introduced;
	none of them is easy to understand.
	Moreover, it turns out that in practice they operate very slowly.
	On the other side there is the Zielonka's recursive algorithm, which is very simple, exponential in the worst case, and the fastest in practice.
	We combine these two approaches: we propose a small modification of the Zielonka's algorithm, which ensures that the running time is at most quasi-polynomial.
	In effect, we obtain a simple algorithm that solves parity games in quasi-polynomial time.
	We also hope that our algorithm, after further optimizations, can lead to an algorithm that shares the good performance of the Zielonka's algorithm on typical inputs,
	while reducing the worst-case complexity on difficult inputs.
\end{abstract}

\section{Introduction}

	The fundamental role of parity games in automata theory, logic, and their applications to verification and synthesis is doubtless,
	hence it is pointless to elaborate on their importance.
	Let us only mention that the algorithmic problem of finding the winner in parity games is polynomial-time equivalent 
	to the emptiness problem for nondeterministic automata on infinite trees with parity acceptance conditions, and to the model-checking problem for modal $\mu$-calculus~\cite{EJS01}.
	It also lies at the heart of algorithmic solutions to the Church's synthesis problem~\cite{RabinBook}. 
	The impact of parity games reaches relatively far areas of computer science, like Markov decision processes~\cite{FearnleyMDP} and linear programming~\cite{FHZ-simplex}.

	It is a long-standing open question whether parity games can be solved in polynomial-time.
	Several results show that they belong to some classes ``slightly above'' polynomial time.
	Namely, deciding the winner of parity games was shown to be in $\mathsf{NP}\cap\mathsf{coNP}$~\cite{EJS01}, and in $\mathsf{UP}\cap\mathsf{coUP}$~\cite{up-co-up},
	while computing winning strategies is in \textsf{PLS}, \textsf{PPAD}, and even in their subclass \textsf{CLS}~\cite{Daskalakis-Papadimitriou}.
	The same holds for other kinds of games: mean-payoff games~\cite{mean-payoff}, discounted games, and simple stochastic games~\cite{stochastic};
	parity games, however, are the easiest among them, in the sense that there are polynomial-time reductions from parity games to the other kinds of games~\cite{up-co-up,mean-payoff},
	but no reductions in the opposite direction are known.

	Describing the algorithmic side of solving parity games, one has to start with the Zielonka's algorithm~\cite{Zielonka},
	being an adaptation of an approach proposed by McNaughton to solve Muller games~\cite{McNaughton}.
	This algorithm consists of a single recursive procedure, being simple and very natural;
	one may say that it computes who wins the game ``directly from the definition''.
	Its running time is exponential in the worst case~\cite{exponential-lower-bound,exponential-lower-bound2,gazda-phd}, but on many typical inputs it works much faster.
	For over two decades researchers were trying to cutback the complexity of solving parity games, which resulted in a series of algorithms, 
	all of which were either exponential~\cite{BCJLM97,Seidl96,old-progress-measure,strategy-improvement,Schewe-big-steps,priority-promotion},
	or mildly subexponential~\cite{randomized-subexponential,subexponential}.
	The next era came unexpectedly in 2017 with a breakthrough result of Calude, Jain, Khoussainov, Li, and Stephan~\cite{calude} (see also~\cite{parity-short,parity-excursion}), 
	who designed an algorithm working in quasi-polynomial time.
	This invoked a series of quasi-polynomial-time algorithms, which appeared soon after~\cite{small-progress-measure,Fearnley,Lehtinen}.
	These algorithms are quite involved (at least compared to the simple recursive algorithm of Zielonka),
	and it is not so trivial to understand them.

	The four quasi-polynomial-time algorithms~\cite{calude,small-progress-measure,Fearnley,Lehtinen}, at first glance being quite different,
	actually proceed along a similar line (as observed by Bojańczyk and Czerwiński~\cite{separation-toolbox} and Czerwiński et al.~\cite{lower-bound}).
	Namely, out of all the four algorithms one can extract a construction of a safety automaton 
	(nondeterministic in the case of Lehtinen~\cite{Lehtinen}, and deterministic in the other algorithms), 
	which accepts all words encoding plays that are decisively won by one of the players 
	(more precisely: plays consistent with some positional winning strategy), 
	and rejects all words encoding plays in which the player loses (for plays that are won by the player, but not decisively, the automaton can behave arbitrarily).
	This automaton does not depend at all on the game graph; it depends only on its size.
	Having an automaton with the above properties, it is not difficult to convert the original parity game into an equivalent safety game
	(by taking a ``product'' of the parity game and the automaton), which can be solved easily---and all the four algorithms actually proceed this way,
	even if it is not stated explicitly than such an automaton is constructed.
	As shown in Czerwiński et al.~\cite{lower-bound}, all automata having the aforementioned properties have to look very similar: their states have to be leaves of some so-called universal tree;
	particular papers propose different constructions of these trees, and of the resulting automata (of quasi-polynomial size).
	Moreover, Czerwiński et al.~\cite{lower-bound} show a quasi-polynomial lower bound for the size of such an automaton.

	In this paper we propose a novel quasi-polynomial-time algorithm solving parity games.
	It is obtained by applying a small modification to the Zielonka's recursive algorithm;
	this modification guarantees that the worst-case running time of this algorithm, being originally exponential, becomes quasi-polynomial.
	The simplicity of the Zielonka's algorithm remains in place;
	we avoid complicated considerations accompanying all the previous quasi-polynomial-time algorithms.
	Another point is that our algorithm exploits the structure of parity games in a rather different way from the four previous quasi-polynomial-time algorithms.
	Indeed, the other algorithms construct automata that are completely independent from a particular game graph given on input---%
	they work in exactly the same way for every game graph of a considered size.
	The behaviour of our algorithm, in contrast, is highly driven by an analysis of the game graph given on input.
	In particular, although our algorithm is not faster than quasi-polynomial, it does not fit to the ``separator approach'' 
	in which a quasi-polynomial lower bound of Czerwiński et al.~\cite{lower-bound} exists.
	
	The running time of our algorithm is quasi-polynomial, and the space complexity is quadratic 
	(more precisely, $O(n\cdot h)$, where $n$ is the number of nodes in the game graph, and $h$ is the maximal priority appearing there).
	
	Let us also mention the practical side of the world.
	It turns out that parity games are one of the areas where theory does not need to meet practice: the quasi-polynomial-time algorithms, although fastest in theory, are actually the slowest.
	The most exhaustive comparison of existing algorithms was performed by Tom van Dijk~\cite{oink}.
	In his Oink tool he has implemented several algorithms, with different optimizations.
	Then, he has evaluated them on a benchmark of Keiren~\cite{benchmark}, containing multiple parity games obtained from model checking and equivalence checking tasks, 
	as well as on different classes of random games.
	It turns out that the classic recursive algorithm of Zielonka~\cite{Zielonka} performs the best, 
	\textit{ex aequo} with the recent priority promotion algorithm~\cite{priority-promotion}.
	After that, we have the strategy improvement algorithm~\cite{strategy-improvement,strategy-improvement-implem}, being a few times slower.
	Far later, we have the small progress measure algorithm~\cite{old-progress-measure}.
	At the very end, with a lot of timeouts, we have the quasi-polynomial-time algorithm of Fearnley, Jain, Schewe, Stephan, and Wojtczak~\cite{Fearnley}.
	The other quasi-polynomial-time algorithms were not implemented due to excessive memory usage.
	
	While developing the current algorithm, we hoped that it will share the good performance with the Zielonka's algorithm, on which it is based.
	Unfortunately, preliminary experiments have shown that this is not necessarily the case.
	It turns out that
	\begin{itemize}
	\item	on random games our algorithm performs similarly to the slowest algorithms implemented in Oink;
	\item	on crafted game families that are difficult for the Zielonka's algorithm, our algorithm is indeed faster from it, but not dramatically faster;
	\item	the only think that is optimistic is that on games with a very low number of priorities our algorithm performs similarly to the fastest algorithms.
	\end{itemize}
	Because the empirical results of a direct implementation of the algorithm are completely unsatisfactory, we do not include a full description of our experiments.
	Instead, we leave an efficient implementation for a future work.
	Beside of the discouraging outcomes, we believe that our idea, via further optimizations, can lead to an algorithm that is both fast in practice and has a good worst-case complexity
	(see the concluding section for more comments).

\section{Preliminaries}

	A parity game is played on a \emph{game graph} between two players, called Even or Odd (shortened sometimes to $E$ and $O$).
	A game graph consists of
	\begin{itemize}
	\item	a directed graph $G$, where we require that every node has at least one successor, and where there are no self-loops (i.e., edges from a node to itself);
	\item	a labeling of every node $v$ of $G$ by a positive natural number $\pi(v)$, called its \emph{priority};
	\item	a partition of nodes of $G$ between nodes owned by Even and nodes owned by Odd.
	\end{itemize}
	An infinite path in $G$ is called a \emph{play}, while a finite path in $G$ is called a \emph{partial play}.
	The game starts in a designated starting node.
	Then, the player to which the current node belongs, selects a successor of this node, and the game continues there.
	In effect, after a finite time a partial play is obtained, and at the end, after infinite time, this results in a play.
	We say that a play $v_1,v_2,\dots$ is winning for Even if $\limsup_{i\to\infty}\pi(v_i)$ is even (i.e., if the maximal priority seen infinitely often is even).
	Conversely, the play is winning for Odd if $\limsup_{i\to\infty}\pi(v_i)$ is odd.
	
	A \emph{strategy} of player $P\in\{\mbox{Even},\mbox{Odd}\}$ is a function that maps every partial play that ends in a node of $P$ to some its successor.
	Such a function says how $P$ will play in every situation of the game (depending on the history of that game).
	When a (partial) play $\pi$ follows a strategy $\sigma$ in every step in which player $P$ is deciding, we say that $\pi$ \emph{agrees} with $\sigma$.
	A strategy $\sigma$ is \emph{winning} for $P$ from a node $v$ if every play that starts in $v$ and agrees with $\sigma$ is winning for $P$.
	While saying ``player $P$ wins from a node $v$'' we usually mean that $P$ has a winning strategy from $v$.
	Let $\Win_P(G)$ be the set of nodes of $G$ from which $P$ wins; it is called the \emph{winning region} of $P$.
	By the Martin's theorem~\cite{Martin-determinacy} we know that parity games are determined: 
	in every game graph $G$, and for every node $v$ of $G$ either Even wins from $v$, or Odd wins from $v$.
	In effect, $\Win_E(G)$ and $\Win_O(G)$ form a partition of the node set of $G$.

	During the analysis, we also consider games with other winning conditions.
	A \emph{winning condition} is a set of plays.
	The winning conditions of Even and Odd considered in parity games are denoted \limsupeven and \limsupodd, respectively.
	Beside of that, for every set $S$ of nodes, let $\safety(S)$ be the set of plays that use only nodes from $S$.
	
	A \emph{dominion} for Even is a set $S$ of nodes such that from every $v\in S$ Even wins the game with the condition $\limsupeven\cap\safety(S)$;
	in other words, from every node of $S$ he can win the parity game without leaving $S$.
	Likewise, a dominion for Odd is a set $S$ of nodes such that from every $v\in S$ Odd wins the game with the condition $\limsupodd\cap\safety(S)$.
	Notice that the whole $\Win_P(G)$ is a dominion for $P$ (where $P\in\{\mbox{Even},\mbox{Odd}\}$).
	Indeed, if Even is going to win from some $v\in\Win_E(G)$, the play cannot leave $\Win_E(G)$ and enter a node $v'\in\Win_O(G)$, 
	as then Odd could use his winning strategy from $v'$ and win the whole game;
	here we use the fact that all suffixes of a play in \limsupeven are also in \limsupeven.
	For $P=\mbox{Odd}$ the situation is symmetric.

\section{Standard Zielonka's Algorithm}

	Before presenting our algorithm, we recall the standard Zielonka's algorithm, as a reference.
	
	For a set of nodes $N$ in a game graph $G$, and for a player $P\in\{\mbox{Even},\mbox{Odd}\}$,
	we define the \emph{attractor} of $N$, denoted $\Atr_P(G,N)$, to be the set of nodes of $G$ from which $P$ can force to reach a node from $N$.
	In other words, $\Atr_P(G,N)$ is the smallest set such that 
	\begin{itemize}
	\item	$N\subseteq\Atr_P(G,N)$,
	\item	if $v$ is a node of $P$ and some its successor is in $\Atr_P(G,N)$, then $v\in\Atr_P(G,N)$, and
	\item	if $v$ is a node of the opponent of $P$ and all its successors are in $\Atr_P(G,N)$, then $v\in\Atr_P(G,N)$.
	\end{itemize}
	Clearly $\Atr_P(G,N)$ can be computed in time proportional to the size of $G$.
	
	\begin{algorithm}
	\caption{Standard Zielonka's Algorithm}\label{alg:std}
	\begin{algorithmic}[1]
	\Procedure[$h$ is an \underline{even} upper bound for priorities in $G$]{Solve$_E$}{$G,h$}
		\DoBegin
			\State $N_h=\{v\in\nodes(G)\mid \pi(v)=h\}$;\Comment{nodes with the highest priority}
			\State $H=G\setminus\Atr_E(G,N_h)$;\Comment{new game: reaching priority $h$ $\to$ win}
			\State $W_O=\textsc{Solve}_O(H,h-1)$;\Comment{in $W_O$ we lose before reaching priority $h$}
			\State $G=G\setminus\Atr_O(G,W_O)$;\Comment{possibly $N_h\cap\Atr_O(G,W_O)\neq\emptyset$}
		\EndDoWhile{$W_O\neq\emptyset$}
	\EndProcedure
	\end{algorithmic}
	\end{algorithm}
	
	Algorithm~\ref{alg:std} is the standard Zielonka's algorithm.
	The procedure $\textsc{Solve}_E(G,h)$ returns $\Win_E(G)$, the winning region of Even,
	if $h$ is an even number that is greater or equal than all priorities appearing in $G$.
	A procedure $\textsc{Solve}_O(G,h)$ is also needed; it is identical to $\textsc{Solve}_E(G,h)$ except that the roles of $E$ and $O$ are swapped; it returns $\Win_O(G)$, the winning region of Odd.
	While writing $G\setminus S$, we mean the game obtained by removing from $G$ all nodes in $S$, and all edges leading to nodes in $S$ or starting from nodes in $S$.
	We use this construct only when $S$ is an attractor; in such a case, if all successors of a node $v$ are removed, then $v$ is also removed 
	(i.e., if all successors of $v$ belong to an attractor, then $v$ belongs to the attractor as well).
	In effect $G\setminus S$ is a valid game graph (every its node has at least one successor).
	
	We remark that the algorithm is presented in a slightly different way than usually.
	Namely, we use here a loop, while the usual presentation does not use a loop but rather calls recursively $\textsc{Solve}_E(G\setminus\Atr_O(G,W_O),h)$ at the end of the procedure.
	This is only a superficial difference in the presentation, but is useful while modifying the algorithm in the next section.

	\begin{figure*}
		\begin{center}
			\import{pics/}{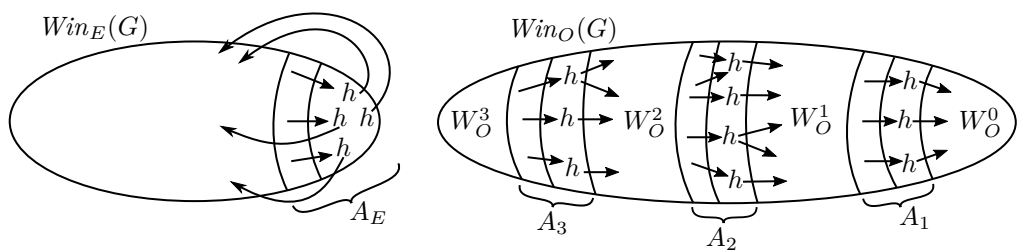}
		\end{center}
		\caption{The structure of winning regions in a parity game}
		\label{fig:1}
	\end{figure*}
	
	The algorithm can be understood while looking at Figure~\ref{fig:1}.
	Let $h$ be the highest priority used in $G$; assume that it is even.
	The game graph $G$ can be divided into two parts: $\Win_E(G)$ and $\Win_O(G)$.
	In $\Win_E(G)$ we can distinguish the attractor of nodes with priority $h$ (denoted $A_E$).
	Odd either loses inside $\Win_E(G)\setminus A_E$, or enters $A_E$, which causes that a node with priority $h$ is seen, and then the game continues in some node of $\Win_E(G)$.
	The winning region of Odd, $\Win_O(G)$, can be divided into multiple parts.
	We have a part $W_O^0$, where Odd can win without seeing a node of priority $h$.
	Then, we have nodes of priority $h$ from which Even is forced to enter $W_O^0$, and their attractor, denoted $A_1$.
	Then, we have a part $W_O^1$, where Odd can ensure that the play is either winning for him inside $W_O^1$ or enters $A_1$;
	in other words, from nodes of $W_O^1$ Odd can win while seeing $h$ at most once.
	We also have parts $W_O^i$ for larger $i$, and corresponding attractors $A_i$.
	
	While running the algorithm, this partition of $G$ is not known, and has to be discovered.
	To this end, the algorithm assumes first (in the game $H$) that all nodes of priority $h$ are winning for Even.
	The first call to $\textsc{Solve}_O(H,h-1)$ returns the set $W_O^0$ of nodes where Odd wins without seeing a node of priority $h$.
	We then remove them from the game, together with the attractor $A_1$.
	In the next step, $\textsc{Solve}_O(H,h-1)$ returns the set $W_O^1$, and so on.
	At the end the whole $\Win_O(G)$ becomes removed, and the procedure returns $\Win_E(G)$.

\section{Quasi-Polynomial-Time Algorithm}
	
	We now present a modification to Algorithm~\ref{alg:std} that results in obtaining quasi-polynomial running time, in the worst case.
	
	The modification can be understood while looking again at Figure~\ref{fig:1}.
	The key observation is that, while $\Win_O(G)$ is of size at most $n$ (where $n$ is the number of nodes in $G$),
	then most of its parts $W_O^i$ are smaller.
	Namely, most of them have to be of size at most $\frac{n}{2}$, and only one of them can be larger than $\frac{n}{2}$.
	We use this observation, and while looking for $W_O^i$, we search for a winning region (for a dominion) of size at most $\frac{n}{2}$.
	Usually this is enough; only once it is not enough: one $W_O^i$ can be larger than $\frac{n}{2}$ and it will not be found if we only look for a set of size at most $\frac{n}{2}$.
	But when the algorithm finds no set of size at most $\frac{n}{2}$, we can once search for $W_O^i$ of an arbitrary size.
	After that, we know that all following sets $W_O^i$ are again of size at most $\frac{n}{2}$.
	While going recursively, we notice that every $W_O^i$ can be further subdivided in a similar way, while splitting on the priority $h-2$.
	If $|W_O^i|\leq\frac{n}{2}$, we again have the property that most of the parts of $W_O^i$ are of size at most $\frac{n}{4}$, and only one of them can be larger than $\frac{n}{4}$.
	
	To exploit this observation, in the recursive calls we pass two precision parameters, $p_E$ and $p_O$ (one for every of the players),
	saying that we search for winning sets of size at most $p_E$ for Even, and at most $p_O$ for Odd.
	The modified procedure is presented as Algorithm~\ref{alg:new}.
	Again, one also needs a procedure $\textsc{Solve}_O$, which is obtained from $\textsc{Solve}_E$ by literally changing every $E$ to $O$ and \textit{vice versa}.

	\begin{algorithm}
	\caption{Quasi-Polynomial-Time Algorithm}\label{alg:new}
	\begin{algorithmic}[1]
	\Procedure[$p_E,p_O$ are new ``precision'' parameters]{Solve$_E$}{$G, h, p_E, p_O$}
		\If {$G=\emptyset\lor p_E\leq 1$} \State \Return $\emptyset$;\Comment{we assume that there are no self-loops in $G$}
		\DoBegin
			\State $N_h=\{v\in\nodes(G)\mid \pi(v)=h\}$;
			\State $H=G\setminus\Atr_E(G,N_h)$;
			\State $W_O=\textsc{Solve}_O(H,h-1,\lfloor p_O/2\rfloor,p_E)$;\Comment{precision decreased}
			\State $G=G\setminus\Atr_O(G,W_O)$;
		\EndDoWhile{$W_O\neq\emptyset$}
		\State $N_h=\{v\in\nodes(G)\mid \pi(v)=h\}$;
		\State $H=G\setminus\Atr_E(G,N_h)$;
		\State $W_O=\textsc{Solve}_O(H,h-1,p_O,p_E)$;\Comment{we try once with the full precision}
		\State $G=G\setminus\Atr_O(G,W_O)$;
		\While{$W_O\neq\emptyset$}
			\State $N_h=\{v\in\nodes(G)\mid \pi(v)=h\}$;
			\State $H=G\setminus\Atr_E(G,N_h)$;
			\State $W_O=\textsc{Solve}_O(H,h-1,\lfloor p_O/2\rfloor,p_E)$;\Comment{again, precision decreased}
			\State $G=G\setminus\Atr_O(G,W_O)$;
		\EndWhile
		\State \Return $\nodes(G)$;
	\EndProcedure
	\end{algorithmic}
	\end{algorithm}

	We start the algorithm with $p_E=p_O=n$, where $n$ is the number of nodes in $G$.
	In the procedure we have now, in a sense, three copies of the previous procedure, corresponding to three stages.
	In the first stage, in lines 5-10, we look for sets $W_O^i$ of size at most $\lfloor\frac{p_O}{2}\rfloor$.
	If the returned set is empty, this may mean that the next $W_O^i$ either is empty, or is of size greater than $\lfloor\frac{p_O}{2}\rfloor$.
	Then, in lines 11-14, we once search for a set $W_O^i$ of size at most $p_O$ (knowing that if it is nonempty, then its size is greater than $\lfloor\frac{p_O}{2}\rfloor$).
	Finally, in the loop in lines 15-20, we again look for sets $W_O^i$ of size at most $\lfloor\frac{p_O}{2}\rfloor$ 
	(because we have already found a set of size greater than $\lfloor\frac{p_O}{2}\rfloor$, all the remaining sets have size at most $\lfloor\frac{p_O}{2}\rfloor$).

\section{Complexity Analysis}
	
	Let us analyze the complexity of our algorithm.
	
	First, we observe that the space complexity is $O(n\cdot h)$, where $n$ is the number of nodes, and $h$ is the maximal priority.
	Indeed, the depth of the recursion is at most $h$, and on every step we only need to remember some sets of nodes.
	
	We now come to the running time.
	As it is anyway worse than the running time of the other quasi-polynomial-time algorithms, 
	we do not aim in proving a very tight upper bound; we only prove that the running time is quasi-polynomial.
	
	Let $R(h,l)$ be the number of (nontrivial) executions of the $\textsc{Solve}_E$ and $\textsc{Solve}_O$ procedures performed during one call to $\textsc{Solve}_E(G,h,p_E,p_O)$ with 
	$\lfloor\log p_E\rfloor+\lfloor\log p_O\rfloor=l$, 
	and with $G$ having at most $n$ nodes (where $n$ is fixed).
	We only count here nontrivial executions, that is, such that do not leave the procedure in line 4.
	Clearly $R(0,l)=R(h,0)=0$.
	For $h,l\geq 1$ it holds that
	\begin{align}
		R(h,l)\leq 1+n\cdot R(h-1,l-1)+R(h-1,l)\,.\label{eq:1}
	\end{align}
	Indeed, in $\textsc{Solve}_E$ after every call to $\textsc{Solve}_O$ we remove at least one node from $G$,
	with the exception of two such calls: the last call in line 8, and the last call ever.
	In effect, in lines 8 and 18 we have at most $n$ calls to $\textsc{Solve}_O$ with decreased precision 
	(plus, potentially, the $(n+1)$-th call with empty $G$, which is not included in $R(h,l)$),
	and in line 13 we have one call to $\textsc{Solve}_O$ with full precision. 
	Notice that $\lfloor\log p_O\rfloor$ (hence also $l$) decreases by $1$ in the decreased-precision call.
	
	Using Inequality~\eqref{eq:1} we now prove by induction that $R(h,l)\leq n^l\cdot\binom{h+l}{l}-1$.
	For $h=0$ and for $l=0$ the inequality holds.
	For $h,l\geq 1$ we have that 
	\begin{align*}
		R(h,l)&\leq 1+n\cdot R(h-1,l-1)+R(h-1,l)\\
		&\leq 1+n\cdot\left(n^{l-1}\cdot\binom{h-1+l-1}{l-1}-1\right)+n^l\cdot\binom{h-1+l}{l}-1\\
		&\leq n^l\cdot\left(\binom{h-1+l}{l-1}+\binom{h-1+l}{l}\right)-1\\
		&=n^l\cdot\binom{h+l}{l}-1\,.
	\end{align*}
	In effect, $R(h,l)\leq n^l\cdot(h+l)^l$.
	Recalling that we start with $l=2\cdot\lfloor\log n\rfloor$, we see that this number is quasi-polynomial in $n$ and $h$.
	This concludes the proof, 
	since obviously a single execution of the $\textsc{Solve}_E$ procedure (not counting the running time of recursive calls) costs polynomial time.

\section{Correctness}

	We now justify correctness of the algorithm.
	This amounts to proving the following lemma.
	
	\begin{lemma}\label{lem:1}
		Procedure $\textsc{Solve}_E(G,h,p_E,p_O)$ returns a set $W_E$ such that for every $S\subseteq\nodes(G)$,
		\begin{itemize}
		\item	if $S$ is a dominion for Even, and $|S|\leq p_E$, then $S\subseteq W_E$, and
		\item	if $S$ is a dominion for Odd, and $|S|\leq p_O$, then $S\cap W_E=\emptyset$.
		\end{itemize}
	\end{lemma}
	
	Notice that in $G$ there may be nodes that do not belong to any dominion smaller than $p_E$ or $p_O$; 
	for such nodes we do not specify whether or not they are contained in $W_E$.
	
	Recall that $\Win_E(G)$ is a dominion for Even, and $\Win_O(G)$ is a dominion for Odd.
	Thus, using Lemma~\ref{lem:1} we can conclude that for $p_E=p_O=n$ the procedure returns $\Win_E(G)$, the winning region of Even.

	One may wonder why we use dominions in the statement of the lemma, instead of simply saying that if $|\Win_E(G)|\leq p_E$, then $\Win_E(G)\subseteq W_E$.
	Such a simplified statement, however, is not suitable for induction.
	Indeed, while switching from the game $G$ to the game $H$ (created in lines 7, 12, 17) the winning regions of Even may increase dramatically,
	because in $H$ Odd is not allowed to visit any node with priority $h$.
	Nevertheless, the winning region of Even in $G$, and any dominion of Even in $G$, remains a dominion in $H$ (when restricted to nodes of $H$).
	
	Before proving Lemma~\ref{lem:1}, let us observe two facts about dominions.
	In their statements $P\in\{\mbox{Even},\mbox{Odd}\}$ is one of the players, and $\overline P$ is his opponent.
	
	\begin{fact}\label{fact:usun-moj}
		If $S$ is a dominion for $P$ in a game $G$, and $X$ is a set of nodes of $G$, then $S\setminus\Atr_P(G,X)$ is a dominion for $P$ in $G\setminus\Atr_P(G,X)$.
	\end{fact}
	
	\begin{proof}
		Denote $S'=S\setminus\Atr_P(G,X)$ and $G'=G\setminus\Atr_P(G,X)$.
		By definition, from every node $v\in S$ player $P$ wins with the condition $\limsupP\cap\safety(S)$ in $G$, using some winning strategy.
		Observe that using the same strategy he wins with the condition $\limsupP\cap\safety(S')$ in $G'$ (assuming that the starting node $v$ is in $S'$).
		The strategy remains valid in $G'$, because every node $u$ of player $P$ that remains in $G'$ has the same successors in $G'$ as in $G$
		(conversely: if some of successors of $u$ belongs to $\Atr_P(G,X)$, then $u$ also belongs to $\Atr_P(G,X)$).		
	\end{proof}
	
	\begin{fact}\label{fact:usun-jego}
		If $S$ is a dominion for $P$ in a game $G$, and $X$ is a set of nodes of $G$ such that $S\cap X=\emptyset$, 
		then $S$ is a dominion for $P$ in $G\setminus\Atr_{\overline P}(G,X)$ (in particular $S\subseteq\nodes(G\setminus\Atr_{\overline P}(G,X))$).
	\end{fact}
	
	\begin{proof}
		Denote $G'=G\setminus\Atr_{\overline P}(G,X)$.
		Suppose that there is some $v\in S\cap\Atr_{\overline P}(G,X)$.
		On the one hand, $P$ can guarantee that, while starting from $v$, the play stays in $S$ (by the definition of a dominion);
		on the other hand, $\overline P$ can force to reach the set $X$ (by the definition of an attractor), which is disjoint from $S$.
		Thus such a node $v$ could not exist, we have $S\subseteq\nodes(G')$.
		
		It remains to observe that from every node $v\in S$ player $P$ wins with the condition $\limsupP\cap\safety(S)$ also in the restricted game $G'$,
		using the same strategy as in $G$.
		Indeed, a play in $G$ following this strategy never leaves $S$, and the whole $S$ remains unchanged in $G'$.
	\end{proof}
	
	We are now ready to prove Lemma~\ref{lem:1}.
	
	\begin{proof}[Proof of Lemma~\ref{lem:1}]
		We prove the lemma by induction on $h$.
		Consider some execution of the procedure.
		By $G^i, N_h^i, H^i, W_O^i$ we denote values of the variables $G, N_h, H, W_O$ just after the $i$-th call to $\textsc{Solve}_O$ in one of the lines 8, 13, 18;
		in lines 9, 14, 19 we create $G^{i+1}$ out of $G^i$ and $W_O^i$.
		In particular $G^1$ equals the original game $G$, and at the end we return $\nodes(G^{m+1})$, where $m$ is the number of calls to $\textsc{Solve}_O$.
	
		Concentrate on the first item of the lemma: fix an Even's dominion $S$ in $G$ (i.e., in $G^1$) such that $|S|\leq p_E$.
		Assume that $S\neq\emptyset$ (for $S=\emptyset$ there is nothing to prove).
		Notice first that a nonempty dominion has at least two nodes (by assumption there are no self-loops in $G$, hence every play has to visit at least two nodes),
		thus, because $S\subseteq\nodes(G)$ and $|S|\leq p_E$, we have that $G\neq\emptyset$ and $p_E>1$.
		It means that the procedure does not return in line 4.
		We thus need to prove that $S\subseteq\nodes(G^{m+1})$.
		
		We actually prove that $S$ is a dominion for Even in $G^i$ for every $i\in\{1,\dots,m+1\}$, meaning in particular that $S\subseteq\nodes(G^i)$.
		This is shown by an internal induction on $i$.
		The base case ($i=1$) holds by assumption.
		For the induction step, consider some $i\in\{1,\dots,m\}$.
		By the induction assumption $S$ is a dominion for Even in $G^i$, and we need to prove that it is a dominion for Even in $G^{i+1}$.
		
		Consider $S^i=S\cap\nodes(H^i)$.
		Because $S^i=S\setminus\Atr_E(G^i,N_h^i)$, by Fact~\ref{fact:usun-moj} the set $S^i$ is a dominion for Even in $H^i=G^i\setminus\Atr_E(G^i,N_h^i)$, and obviously $|S^i|\leq|S|\leq p_E$.
		By the assumption of the external induction (which can be applied to $\textsc{Solve}_O$, by symmetry) 
		it follows that $S^i\cap W_O^i=\emptyset$, so also $S\cap W_O^i=\emptyset$ (because $W_O^i$ contains only nodes of $G^i$, while $S\setminus S^i$ contains no nodes of $G^i$).
		Thus, by Fact~\ref{fact:usun-jego} the set $S$ is a dominion for Even in $G^{i+1}=G^i\setminus\Atr_O(G^i,W_O^i)$.
		This finishes the proof of the first item.
	
		Now we prove the second item of the lemma.
		To this end, fix some Odd's dominion $S$ in $G$ such that $|S|\leq p_O$.
		If $p_E\leq 1$, we return $W_E=\emptyset$ (line 4), so clearly $S\cap W_E=\emptyset$.
		The interesting case is when $p_E\geq 2$.
		Denote $S^i=S\cap\nodes(G^i)$ for all $i\in\{1,\dots,m+1\}$; we first prove that $S^i$ is a dominion for Odd in $G^i$.
		This is shown by induction on $i$.
		The base case of $i=1$ holds by assumption, because $G^1=G$ and $S^1=S$.
		For the induction step, assume that $S^i$ is a dominion for Odd in $G^i$, for some $i\in\{1,\dots,m\}$.
		By definition $G^{i+1}=G^i\setminus\Atr_O(G^i,W_O^i)$ and $S^{i+1}=S^i\setminus\Atr_O(G^i,W_O^i)$,
		so $S^{i+1}$ is a dominion for Odd in $G^{i+1}$ by Fact~\ref{fact:usun-moj}, which finishes the inductive proof.
		
		For $i\in\{1,\dots,m\}$, let $Z^i$ be the set of nodes (in $S^i\setminus N_h^i$) from which Odd wins with the condition $\limsupodd\cap\safety(S^i\setminus N_h^i)$ in $G^i$
		(that is, where Odd can win without seeing priority $h$---the highest even priority).
		Let us observe that if $S^i\neq\emptyset$ then $Z^i\neq\emptyset$ ($\clubsuit$).
		Indeed, suppose to the contrary that $Z^i=\emptyset$, 
		and consider an Odd's strategy that allows him to win with the condition $\limsupodd\cap\safety(S^i)$ in $G^i$, from some node $v_0\in S^i$.
		Because $v_0\not\in Z^i$, this strategy in not winning for the condition $\limsupodd\cap\safety(S^i\setminus N_h^i)$, 
		so Even, while playing against this strategy, can reach a node $v_1$ in $N_h^i$ (as he cannot violate the parity condition nor leave $S^i$).
		For the same reason, because $v_1\not\in Z^i$, Even can continue and reach a node $v_2$ in $N_h^i$.
		Repeating this forever, Even gets priority $h$ (which is even and is the highest priority) infinitely many times, contradicting the fact that the strategy was winning for Odd.
		
		Observe also that from nodes of $Z^i$ Odd can actually win with the condition $\limsupodd\cap\safety(Z^i)$ in $G^i$,
		using the strategy that allows him to win with the condition $\limsupodd\cap\safety(S^i\setminus N_h^i)$.
		Indeed, if a play following this strategy enters some node $v$, then from this node $v$ Odd can still win with the condition $\limsupodd\cap\safety(S^i\setminus N_h^i)$,
		which means that these nodes belongs to $Z^i$.
		It follows that $Z^i$ is a dominion for Odd in $G^i$.
		Moreover, because $Z^i\cap N_h^i=\emptyset$, from Fact~\ref{fact:usun-jego} we have that $Z^i$ is a dominion for Odd in $H^i=G^i\setminus\Atr_E(G^i,N_h^i)$.
		
		Let $k$ be the number of the call to $\textsc{Solve}_O$ that is performed in line 13 (calls number $1,\dots,k-1$ are performed in line 8, and calls number $k+1,\dots,m$ are performed in line 18).
		Recall that $W_O^i$ is the set returned by a call to $\textsc{Solve}_O(H^i,h-1,p_O^i,p_E)$, where $p_O^k=p_O$, and $p_O^i=\lfloor\frac{p_O}{2}\rfloor$ if $i\neq k$.
		From the assumption of the external induction, if $|Z^i|\leq\lfloor\frac{p_O}{2}\rfloor$ or if $i=k$ (since $Z^i\subseteq S^i\subseteq S$ and $|S|\leq p_O$, clearly $|Z^i|\leq p_O$), 
		we obtain that $Z^i\subseteq W_O^i$ ($\spadesuit$).
		
		We now prove that $|S^{k+1}|\leq\lfloor\frac{p_O}{2}\rfloor$.
		This clearly holds if $S^{k-1}=\emptyset$, because $S^{k+1}\subseteq S^k\subseteq S^{k-1}$.
		Suppose thus that $S^{k-1}\neq\emptyset$.
		Then $Z^{k-1}\neq\emptyset$, by ($\clubsuit$).
		On the other hand, $W_O^{k-1}=\emptyset$, because we are just about to leave the loop in lines 5-10 (the $k$-th call to $\textsc{Solve}_O$ is in line 13).
		By ($\spadesuit$), if $|Z^{k-1}|\leq\lfloor\frac{p_O}{2}\rfloor$, then $Z^{k-1}\subseteq W_O^{k-1}$, which does not hold in our case.
		Thus $|Z^{k-1}|>\lfloor\frac{p_O}{2}\rfloor$.
		Because $W_O^{k-1}=\emptyset$, we simply have $G^k=G^{k-1}$, and $S^k=S^{k-1}$, and $Z^k=Z^{k-1}$.
		Using ($\spadesuit$) for $i=k$, we obtain that $Z^k\subseteq W_O^k$, and because $S^{k+1}=S^k\setminus\Atr_O(G^k,W_O^k)\subseteq S^k\setminus W_O^k\subseteq S^k\setminus Z^k$
		we obtain that $|S^{k+1}|\leq|S^k|-|Z^k|\leq p_O-(\lfloor\frac{p_O}{2}\rfloor+1)\leq\lfloor\frac{p_O}{2}\rfloor$, as initially claimed.
		
		If $k=m$, we have $Z^m\subseteq W_O^m$ by ($\spadesuit$).
		If $k+1\leq m$, we have $S^m\subseteq S^{k+1}$ (our procedure only removes nodes from the game) and $Z^m\subseteq S^m$, 
		so $|Z^m|\leq\lfloor\frac{p_O}{2}\rfloor$ by the above paragraph,
		and also $Z^m\subseteq W_O^m$ by ($\spadesuit$).
		Because after the $m$-th call to $\textsc{Solve}_O$ the procedure ends, we have $W_O^m=\emptyset$, so also $Z^m=\emptyset$, and thus $S^m=\emptyset$ by ($\clubsuit$).
		We have $S^{m+1}\subseteq S^m$, so $S^{m+1}=S\cap\nodes(G^{m+1})=\emptyset$.
		This is exactly the conclusion of the lemma, since the set returned by the procedure is $\nodes(G^{m+1})$.
	\end{proof}

\section{Conclusions}
	
	To the list of the four existing quasi-polynomial-time algorithms solving parity games, we have added a new one.
	It uses a rather different approach: it analyses recursively the game graph, like the Zielonka's algorithm.

	Notice that the number of recursive calls in our algorithm may be smaller than in the original Zielonka's algorithm, because of the precision parameters, but it may also be larger.
	Indeed, while $\textsc{Solve}_E$ in the original Zielonka's algorithm stops after the first time when a recursive call returns $\emptyset$,
	in our algorithm the procedure stops after the second time when a recursive call returns $\emptyset$.

	The algorithm, as is, turns out not to be very efficient in practice.
	Beside of that, we believe that it can serve as a good starting point for a more optimized algorithm.
	Over the years, some optimizations to the Zielonka’s algorithm were proposed.
	For example, Liu, Duan, and Tian~\cite{LiuDT14} replace the loop guard $W_O=\emptyset$ by $W_O=\Atr_O(G,W_O)$ 
	(which ensures that $W_O$ will be empty in the next iteration of the loop).
	Verver~\cite{verver} proposes to check whether $\Atr_E(G,N_h)$ contains all nodes of priority $h-1$, 
	and if so, to extend $N_h$ by nodes of the next highest Even priority (i.e., $h-2$).
	It seems that these optimizations can be applied to our algorithm as well.
	
	A straightforward optimization is to decrease $p_O$ and $p_E$ to $|G|$ at the beginning of every recursive call.
	
	Another idea is to extend the recursive procedure so that it will return also a Boolean value saying whether the returned set surely equals the whole winning region 
	(i.e., whether the precision parameters have not restricted anything).
	If while making the recursive call with smaller precision (line 8) the answer is positive, but the returned set $W_O$ is empty, 
	we can immediately stop the procedure, without making the recursive call with the full precision (line 13).
	
	One can also observe that the call to $\textsc{Solve}_O$ in line 13 (with the full precision) gets the same subgame $H$ as the last call to $\textsc{Solve}_O$ in line 8 (with decreased precision).
	A very rough idea is to make some use of the computations performed by the decreased-precision call during the full-precision call.
	
	We leave implementation and evaluation of the above (and potentially some other) optimizations for a future work.
	
\bibliography{bib}

\begin{thebibliography}{10}

\bibitem{exponential-lower-bound2}
Massimo Benerecetti, Daniele Dell'Erba, and Fabio Mogavero.
\newblock Robust exponential worst cases for divide-et-impera algorithms for
  parity games.
\newblock In Patricia Bouyer, Andrea Orlandini, and Pierluigi~San Pietro,
  editors, {\em Proceedings Eighth International Symposium on Games, Automata,
  Logics and Formal Verification, GandALF 2017, Roma, Italy, 20-22 September
  2017.}, volume 256 of {\em {EPTCS}}, pages 121--135, 2017.
\newblock \href {http://dx.doi.org/10.4204/EPTCS.256.9}
  {\path{doi:10.4204/EPTCS.256.9}}.

\bibitem{priority-promotion}
Massimo Benerecetti, Daniele Dell'Erba, and Fabio Mogavero.
\newblock Solving parity games via priority promotion.
\newblock {\em Formal Methods in System Design}, 52(2):193--226, 2018.
\newblock \href {http://dx.doi.org/10.1007/s10703-018-0315-1}
  {\path{doi:10.1007/s10703-018-0315-1}}.

\bibitem{randomized-subexponential}
Henrik Bj{\"{o}}rklund and Sergei~G. Vorobyov.
\newblock A combinatorial strongly subexponential strategy improvement
  algorithm for mean payoff games.
\newblock {\em Discrete Applied Mathematics}, 155(2):210--229, 2007.
\newblock \href {http://dx.doi.org/10.1016/j.dam.2006.04.029}
  {\path{doi:10.1016/j.dam.2006.04.029}}.

\bibitem{separation-toolbox}
Mikołaj Boja\'nczyk and Wojciech Czerwi\'nski.
\newblock An automata toolbox, February 2018.
\newblock URL:
  \url{https://www.mimuw.edu.pl/~bojan/papers/toolbox-reduced-feb6.pdf}.

\bibitem{BCJLM97}
Anca Browne, Edmund~M. Clarke, Somesh Jha, David~E. Long, and Wilfredo~R.
  Marrero.
\newblock An improved algorithm for the evaluation of fixpoint expressions.
\newblock {\em Theor. Comput. Sci.}, 178(1-2):237--255, 1997.
\newblock \href {http://dx.doi.org/10.1016/S0304-3975(96)00228-9}
  {\path{doi:10.1016/S0304-3975(96)00228-9}}.

\bibitem{calude}
Cristian~S. Calude, Sanjay Jain, Bakhadyr Khoussainov, Wei Li, and Frank
  Stephan.
\newblock Deciding parity games in quasipolynomial time.
\newblock In Hamed Hatami, Pierre McKenzie, and Valerie King, editors, {\em
  Proceedings of the 49th Annual {ACM} {SIGACT} Symposium on Theory of
  Computing, {STOC} 2017, Montreal, QC, Canada, June 19-23, 2017}, pages
  252--263. {ACM}, 2017.
\newblock \href {http://dx.doi.org/10.1145/3055399.3055409}
  {\path{doi:10.1145/3055399.3055409}}.

\bibitem{stochastic}
Anne Condon.
\newblock The complexity of stochastic games.
\newblock {\em Inf. Comput.}, 96(2):203--224, 1992.
\newblock \href {http://dx.doi.org/10.1016/0890-5401(92)90048-K}
  {\path{doi:10.1016/0890-5401(92)90048-K}}.

\bibitem{lower-bound}
Wojciech Czerwiński, Laure Daviaud, Nathana{\"{e}}l Fijalkow, Marcin
  Jurdziński, Ranko Lazić, and Paweł Parys.
\newblock Universal trees grow inside separating automata: Quasi-polynomial
  lower bounds for parity games.
\newblock In Timothy~M. Chan, editor, {\em Proceedings of the Thirtieth Annual
  {ACM-SIAM} Symposium on Discrete Algorithms, {SODA} 2019, San Diego,
  California, USA, January 6-9, 2019}, pages 2333--2349. {SIAM}, 2019.
\newblock \href {http://dx.doi.org/10.1137/1.9781611975482.142}
  {\path{doi:10.1137/1.9781611975482.142}}.

\bibitem{Daskalakis-Papadimitriou}
Constantinos Daskalakis and Christos~H. Papadimitriou.
\newblock Continuous local search.
\newblock In Dana Randall, editor, {\em Proceedings of the Twenty-Second Annual
  {ACM-SIAM} Symposium on Discrete Algorithms, {SODA} 2011, San Francisco,
  California, USA, January 23-25, 2011}, pages 790--804. {SIAM}, 2011.
\newblock \href {http://dx.doi.org/10.1137/1.9781611973082.62}
  {\path{doi:10.1137/1.9781611973082.62}}.

\bibitem{EJS01}
E.~Allen Emerson, Charanjit~S. Jutla, and A.~Prasad Sistla.
\newblock On model checking for the {\(\mathrm{\mu}\)}-calculus and its
  fragments.
\newblock {\em Theor. Comput. Sci.}, 258(1-2):491--522, 2001.
\newblock \href {http://dx.doi.org/10.1016/S0304-3975(00)00034-7}
  {\path{doi:10.1016/S0304-3975(00)00034-7}}.

\bibitem{FearnleyMDP}
John Fearnley.
\newblock Exponential lower bounds for policy iteration.
\newblock In Samson Abramsky, Cyril Gavoille, Claude Kirchner, Friedhelm~Meyer
  auf~der Heide, and Paul~G. Spirakis, editors, {\em Automata, Languages and
  Programming, 37th International Colloquium, {ICALP} 2010, Bordeaux, France,
  July 6-10, 2010, Proceedings, Part {II}}, volume 6199 of {\em Lecture Notes
  in Computer Science}, pages 551--562. Springer, 2010.
\newblock \href {http://dx.doi.org/10.1007/978-3-642-14162-1_46}
  {\path{doi:10.1007/978-3-642-14162-1_46}}.

\bibitem{strategy-improvement-implem}
John Fearnley.
\newblock Efficient parallel strategy improvement for parity games.
\newblock In Rupak Majumdar and Viktor Kuncak, editors, {\em Computer Aided
  Verification - 29th International Conference, {CAV} 2017, Heidelberg,
  Germany, July 24-28, 2017, Proceedings, Part {II}}, volume 10427 of {\em
  Lecture Notes in Computer Science}, pages 137--154. Springer, 2017.
\newblock \href {http://dx.doi.org/10.1007/978-3-319-63390-9_8}
  {\path{doi:10.1007/978-3-319-63390-9_8}}.

\bibitem{Fearnley}
John Fearnley, Sanjay Jain, Sven Schewe, Frank Stephan, and Dominik Wojtczak.
\newblock An ordered approach to solving parity games in quasi polynomial time
  and quasi linear space.
\newblock In Hakan Erdogmus and Klaus Havelund, editors, {\em Proceedings of
  the 24th {ACM} {SIGSOFT} International {SPIN} Symposium on Model Checking of
  Software, Santa Barbara, CA, USA, July 10-14, 2017}, pages 112--121. {ACM},
  2017.
\newblock \href {http://dx.doi.org/10.1145/3092282.3092286}
  {\path{doi:10.1145/3092282.3092286}}.

\bibitem{exponential-lower-bound}
Oliver Friedmann.
\newblock Recursive algorithm for parity games requires exponential time.
\newblock {\em {RAIRO} - Theor. Inf. and Applic.}, 45(4):449--457, 2011.
\newblock \href {http://dx.doi.org/10.1051/ita/2011124}
  {\path{doi:10.1051/ita/2011124}}.

\bibitem{FHZ-simplex}
Oliver Friedmann, Thomas~Dueholm Hansen, and Uri Zwick.
\newblock Subexponential lower bounds for randomized pivoting rules for the
  simplex algorithm.
\newblock In Lance Fortnow and Salil~P. Vadhan, editors, {\em Proceedings of
  the 43rd {ACM} Symposium on Theory of Computing, {STOC} 2011, San Jose, CA,
  USA, 6-8 June 2011}, pages 283--292. {ACM}, 2011.
\newblock \href {http://dx.doi.org/10.1145/1993636.1993675}
  {\path{doi:10.1145/1993636.1993675}}.

\bibitem{gazda-phd}
Maciej Gazda.
\newblock {\em Fixpoint Logic, Games, and Relations of Consequence}.
\newblock PhD thesis, Eindhoven University of Technology, 2016.
\newblock URL: \url{https://pure.tue.nl/ws/files/16681817/20160315_Gazda.pdf}.

\bibitem{parity-short}
Hugo Gimbert and Rasmus Ibsen{-}Jensen.
\newblock A short proof of correctness of the quasi-polynomial time algorithm
  for parity games.
\newblock {\em CoRR}, abs/1702.01953, 2017.
\newblock \href {http://arxiv.org/abs/1702.01953} {\path{arXiv:1702.01953}}.

\bibitem{up-co-up}
Marcin Jurdziński.
\newblock Deciding the winner in parity games is in {UP} $\cap$ co-{UP}.
\newblock {\em Inf. Process. Lett.}, 68(3):119--124, 1998.
\newblock \href {http://dx.doi.org/10.1016/S0020-0190(98)00150-1}
  {\path{doi:10.1016/S0020-0190(98)00150-1}}.

\bibitem{old-progress-measure}
Marcin Jurdziński.
\newblock Small progress measures for solving parity games.
\newblock In Horst Reichel and Sophie Tison, editors, {\em {STACS} 2000, 17th
  Annual Symposium on Theoretical Aspects of Computer Science, Lille, France,
  February 2000, Proceedings}, volume 1770 of {\em Lecture Notes in Computer
  Science}, pages 290--301. Springer, 2000.
\newblock \href {http://dx.doi.org/10.1007/3-540-46541-3_24}
  {\path{doi:10.1007/3-540-46541-3_24}}.

\bibitem{small-progress-measure}
Marcin Jurdziński and Ranko Lazić.
\newblock Succinct progress measures for solving parity games.
\newblock In {\em 32nd Annual {ACM/IEEE} Symposium on Logic in Computer
  Science, {LICS} 2017, Reykjavik, Iceland, June 20-23, 2017}, pages 1--9.
  {IEEE} Computer Society, 2017.
\newblock \href {http://dx.doi.org/10.1109/LICS.2017.8005092}
  {\path{doi:10.1109/LICS.2017.8005092}}.

\bibitem{subexponential}
Marcin Jurdziński, Mike Paterson, and Uri Zwick.
\newblock A deterministic subexponential algorithm for solving parity games.
\newblock {\em {SIAM} J. Comput.}, 38(4):1519--1532, 2008.
\newblock \href {http://dx.doi.org/10.1137/070686652}
  {\path{doi:10.1137/070686652}}.

\bibitem{benchmark}
Jeroen J.~A. Keiren.
\newblock Benchmarks for parity games.
\newblock In Mehdi Dastani and Marjan Sirjani, editors, {\em Fundamentals of
  Software Engineering - 6th International Conference, {FSEN} 2015 Tehran,
  Iran, April 22-24, 2015, Revised Selected Papers}, volume 9392 of {\em
  Lecture Notes in Computer Science}, pages 127--142. Springer, 2015.
\newblock \href {http://dx.doi.org/10.1007/978-3-319-24644-4_9}
  {\path{doi:10.1007/978-3-319-24644-4_9}}.

\bibitem{parity-excursion}
Bakhadyr Khoussainov.
\newblock A brief excursion to parity games.
\newblock In Mizuho Hoshi and Shinnosuke Seki, editors, {\em Developments in
  Language Theory - 22nd International Conference, {DLT} 2018, Tokyo, Japan,
  September 10-14, 2018, Proceedings}, volume 11088 of {\em Lecture Notes in
  Computer Science}, pages 24--35. Springer, 2018.
\newblock \href {http://dx.doi.org/10.1007/978-3-319-98654-8_3}
  {\path{doi:10.1007/978-3-319-98654-8_3}}.

\bibitem{Lehtinen}
Karoliina Lehtinen.
\newblock A modal {\(\mu\)} perspective on solving parity games in
  quasi-polynomial time.
\newblock In Anuj Dawar and Erich Gr{\"{a}}del, editors, {\em Proceedings of
  the 33rd Annual {ACM/IEEE} Symposium on Logic in Computer Science, {LICS}
  2018, Oxford, UK, July 09-12, 2018}, pages 639--648. {ACM}, 2018.
\newblock \href {http://dx.doi.org/10.1145/3209108.3209115}
  {\path{doi:10.1145/3209108.3209115}}.

\bibitem{LiuDT14}
Yao Liu, Zhenhua Duan, and Cong Tian.
\newblock An improved recursive algorithm for parity games.
\newblock In {\em 2014 Theoretical Aspects of Software Engineering Conference,
  {TASE} 2014, Changsha, China, September 1-3, 2014}, pages 154--161. {IEEE}
  Computer Society, 2014.
\newblock \href {http://dx.doi.org/10.1109/TASE.2014.24}
  {\path{doi:10.1109/TASE.2014.24}}.

\bibitem{Martin-determinacy}
Donald~A. Martin.
\newblock Borel determinacy.
\newblock {\em The Annals of Mathematics}, 102(2):363--371, 1975.

\bibitem{McNaughton}
Robert McNaughton.
\newblock Infinite games played on finite graphs.
\newblock {\em Ann. Pure Appl. Logic}, 65(2):149--184, 1993.
\newblock \href {http://dx.doi.org/10.1016/0168-0072(93)90036-D}
  {\path{doi:10.1016/0168-0072(93)90036-D}}.

\bibitem{RabinBook}
Michael~Oser Rabin.
\newblock {\em Automata on Infinite Objects and Church's Problem}.
\newblock American Mathematical Society, Boston, MA, USA, 1972.

\bibitem{Schewe-big-steps}
Sven Schewe.
\newblock Solving parity games in big steps.
\newblock {\em J. Comput. Syst. Sci.}, 84:243--262, 2017.
\newblock \href {http://dx.doi.org/10.1016/j.jcss.2016.10.002}
  {\path{doi:10.1016/j.jcss.2016.10.002}}.

\bibitem{Seidl96}
Helmut Seidl.
\newblock Fast and simple nested fixpoints.
\newblock {\em Inf. Process. Lett.}, 59(6):303--308, 1996.
\newblock \href {http://dx.doi.org/10.1016/0020-0190(96)00130-5}
  {\path{doi:10.1016/0020-0190(96)00130-5}}.

\bibitem{oink}
Tom van Dijk.
\newblock Oink: An implementation and evaluation of modern parity game solvers.
\newblock In Dirk Beyer and Marieke Huisman, editors, {\em Tools and Algorithms
  for the Construction and Analysis of Systems - 24th International Conference,
  {TACAS} 2018, Held as Part of the European Joint Conferences on Theory and
  Practice of Software, {ETAPS} 2018, Thessaloniki, Greece, April 14-20, 2018,
  Proceedings, Part {I}}, volume 10805 of {\em Lecture Notes in Computer
  Science}, pages 291--308. Springer, 2018.
\newblock \href {http://dx.doi.org/10.1007/978-3-319-89960-2_16}
  {\path{doi:10.1007/978-3-319-89960-2_16}}.

\bibitem{verver}
Maks Verver.
\newblock Practical improvements to parity game solving.
\newblock Master's thesis, University of Twente, 2013.
\newblock URL:
  \url{http://essay.utwente.nl/64985/1/practical-improvements-to-parity-game-solving.pdf}.

\bibitem{strategy-improvement}
Jens V{\"{o}}ge and Marcin Jurdziński.
\newblock A discrete strategy improvement algorithm for solving parity games.
\newblock In E.~Allen Emerson and A.~Prasad Sistla, editors, {\em Computer
  Aided Verification, 12th International Conference, {CAV} 2000, Chicago, IL,
  USA, July 15-19, 2000, Proceedings}, volume 1855 of {\em Lecture Notes in
  Computer Science}, pages 202--215. Springer, 2000.
\newblock \href {http://dx.doi.org/10.1007/10722167_18}
  {\path{doi:10.1007/10722167_18}}.

\bibitem{Zielonka}
Wiesław Zielonka.
\newblock Infinite games on finitely coloured graphs with applications to
  automata on infinite trees.
\newblock {\em Theor. Comput. Sci.}, 200(1-2):135--183, 1998.
\newblock \href {http://dx.doi.org/10.1016/S0304-3975(98)00009-7}
  {\path{doi:10.1016/S0304-3975(98)00009-7}}.

\bibitem{mean-payoff}
Uri Zwick and Mike Paterson.
\newblock The complexity of mean payoff games on graphs.
\newblock {\em Theor. Comput. Sci.}, 158(1{\&}2):343--359, 1996.
\newblock \href {http://dx.doi.org/10.1016/0304-3975(95)00188-3}
  {\path{doi:10.1016/0304-3975(95)00188-3}}.

\end{thebibliography}

\end{document}